\documentclass[12pt, a4paper]{article}
\usepackage{amsmath, amssymb, amsthm}
\usepackage{graphicx}
\usepackage{hyperref}
\usepackage{listings}
\usepackage{xcolor}
\usepackage{tikz}
\usetikzlibrary{automata, positioning, arrows}
\usepackage{float}
\usepackage{authblk}
\usepackage{algorithmicx}
\usepackage{algorithm}
\usepackage{algpseudocode}

\title{Equivalence of Halting Problem to Convergence of Power Series}
\author[1]{Antonio Joaquim Fernandes}
\affil[1]{FCT, Universidade de Cabo Verde}
\date{\today}

\newtheorem{theorem}{Theorem}

\newtheorem{definition}{Definition}

\newtheorem{corollary}{Corollary}

\begin{document}

\maketitle

\begin{abstract}
This paper establishes an equivalence between the halting problem in computability theory and the convergence of power series in mathematical analysis. We demonstrate that for any given computer program, one can algorithmically construct a power series whose convergence behavior is equivalent to the program’s halting status. The construction is explicit and leverages the program's state transitions to define the series' coefficients. Specifically, we show that a power series that converges everywhere corresponds to a non-halting program, while a power series that diverges everywhere except at the origin corresponds to a halting program. This result reveals deep connections between computation and analysis and provides insights into the undecidability of convergence problems in mathematics. While the undecidability of convergence problems was known from previous work, our specific construction provides a particularly clear and direct demonstration of this phenomenon. Furthermore, it opens avenues for exploring the boundaries of decidability within classical analysis.
\end{abstract}

\section{Introduction}
The Halting problem, proven undecidable by Alan Turing in 1936 \cite{turing1936computable}, represents a cornerstone of computability theory. It establishes that no general algorithm exists to determine whether an arbitrary computer program will eventually halt or run indefinitely. 

In mathematical analysis, determining the convergence of power series is a classical problem with no universal decision procedure, despite the existence of various convergence tests \cite{rice1953classes,pourel1989computable, weihrauch2000computable}. The connection between these two fundamental concepts has been explored in various forms in the literature, though our specific construction appears to be novel.

This paper bridges these two fundamental concepts by demonstrating their essential equivalence. We show that the Halting Problem can be reduced to determining the convergence of appropriately constructed power series, and vice versa. This connection provides an analytic interpretation of computability and highlights the computational complexity inherent in analytical problems.

The connection between computability theory and analysis has been explored by several researchers. The most relevant previous work includes:

\subsection{Computable Analysis}
The field of computable analysis, pioneered by mathematicians such as Grzegorczyk \cite{grzegorczyk1955computable}, Lacombe \cite{lacombe1955extension}, and later developed by Pour-El and Richards \cite{pourel1989computable}, established foundations for understanding which analytical objects are computable. These works demonstrated that many problems in analysis are undecidable or non-computable.

\subsection{Specker Sequences}
Ernst Specker \cite{specker1949nicht} constructed the first examples of computable monotone bounded sequences that do not converge to a computable real number. This result demonstrated early connections between computability and convergence.

\subsection{Rice's Theorem and Its Implications}
Rice's Theorem \cite{rice1953classes} established that all non-trivial semantic properties of programs are undecidable. This result has implications for convergence problems, as it suggests that determining whether a computable sequence converges is undecidable in general.

\subsection{Recent Work}
More recently, researchers such as Weihrauch \cite{weihrauch2000computable} and Brattka \cite{brattka2008computable} have developed a more rigorous framework of computable analysis, providing sophisticated tools for understanding the computability of analytical problems.

While these works established important foundations, our specific construction showing a direct equivalence between the Halting Problem and power series convergence appears to be novel in its formulation and simplicity.

\section{Preliminaries}

\subsection{Decidability of formal systems}

A mathematical problem is considered decidable if there exists an effective method (a recursive function or a Turing machine) that can correctly determine the answer to all instances of the problem within a finite number of steps. 

While elementary arithmetic operations and many algebraic problems are decidable, more complex mathematical domains often contain undecidable problems which arises not from insufficient mathematical knowledge but from fundamental computational limitations. Thus, decidability serves as a crucial boundary separating computationally tractable problems from those that fundamentally resist algorithmic solution, regardless of mathematical sophistication.

Kurt Gödel's incompleteness theorems \cite{godel1931} were foundational to mathematical logic and demonstrated limitations of formal systems but the direct proof of the undecidability of the Halting Problem belongs to Church and Turing.

The paradigmatic example of an undecidable problem is Hilbert's Entscheidungsproblem (decision problem), asking whether there exists an algorithm that can determine the validity of any given first-order logic statement. Alonzo Church \cite{church1936note}, in 1936, used lambda calculus to show that no general procedure exists to decide first-order logical validity.

Alan Turing \cite{turing1936computable}, in the same year, independently proved the undecidability of the Entscheidungsproblem using his conceptualization of Turing machines and demonstrating the unsolvability of the Halting Problem, that no general procedure can determine whether an arbitrary computer program will eventually halt or run indefinitely. 

This dual proof established that in Programming as well as in formal framework of mathematical logic, there exist well-formed questions that cannot be resolved by any mechanical procedure.

\subsection{Computability Theory}
While it is also possible to use lambda calculus, we adopt the Turing machine model of computation. 
\begin{definition}
A program $P$ is any finite mathematical object that specifies a partial function $f: \mathbb{N} \to \mathbb{N}$ (or between other countable sets) such that:
\begin{itemize}
\item The program can be effectively encoded as a natural number (Gödel numbering)
\item There exists a universal program/interpreter $U$ such that for any program $P$ and input $x$:
\[
U(P, x) = f_P(x)
\]
where $f_P$ is the function computed by program $P$
\item The set of all programs $\mathcal{P}$ is recursively enumerable
\end{itemize}
\end{definition}
Let $\mathcal{P}$ denote the set of all programs (assuming that the set $\mathcal{P}$ is not itself a program). For a program $P \in \mathcal{P}$ and input $x$, we write $P(x)\downarrow$ if $P$ halts on input $x$, and $P(x)\uparrow$ if it does not.

\begin{definition}[Halting Problem]
The Halting Problem is the decision problem that determines, given a program $P$ and an input $x$, whether $P(x)\downarrow$.
\end{definition}

\begin{theorem}[Turing \cite{turing1936computable}]
The Halting Problem is undecidable.
\end{theorem}

\subsection{Power Series and Effective Convergence}

\begin{definition} A power series, with center at $c=0$, is a formal expression of the type:
\[
S(z) = \sum_{n=0}^{\infty} a_n z^n,
\]
where $a_n$ are complex coefficients and $z$ is a complex variable.
\end{definition}

\begin{definition} A power series $S(z)$ \emph{converges} at $z_0$ if the sequence of partial sums $S_N(z_0) = \sum_{n=0}^{N} a_n z_0^n$ converges as $N \to \infty$.
\end{definition}

The radius of convergence $R$ is given by:
\[
R = \frac{1}{\limsup_{n \to \infty} |a_n|^{1/n}}.
\]

The series converges absolutely for $|z| < R$ and diverges for $|z| > R$.

\begin{definition} A function $f: \mathbb{N} \to \mathbb{N}$ is \emph{computable} if there exists a Turing machine (or equivalent computational model) that, given input $n \in \mathbb{N}$, halts and outputs $f(n)$.
\end{definition}

\begin{definition}
A real number $x$ is computable if there exists a computable function $f: \mathbb{N} \to \mathbb{Q}$ such that for all $n \in \mathbb{N}$:
\[
|f(n) - x| < 2^{-n}
\]
That is, there exists an algorithm that can approximate $x$ to within any desired precision.

Equivalently, a real number $x$ is computable if its decimal (or binary) expansion is computable, meaning there exists an algorithm that, given $n$, outputs the $n$-th digit of $x$.

\end{definition}

\begin{definition} A sequence $\{s_n\}$ of real numbers \emph{converges effectively} to a limit $L$ if there exists a computable function $e: \mathbb{N} \to \mathbb{N}$ such that for all $n \in \mathbb{N}$:
\[
k \geq e(n) \Rightarrow |s_k - L| < 2^{-n}
\]
That is, the rate of convergence is computably bounded.
\end{definition}

\begin{definition} A sequence of functions $\{f_n\}$ \emph{converges effectively uniformly} to $f$ on a domain $D$ if there exists a computable function $e: \mathbb{N} \to \mathbb{N}$ such that for all $n \in \mathbb{N}$ and all $z \in D$:
\[
k \geq e(n) \Rightarrow |f_k(z) - f(z)| < 2^{-n}
\]
\end{definition}

\begin{definition}A power series $S(z) = \sum_{n=0}^{\infty} a_n z^n$ is \emph{effectively computable} on a domain $D$ if:
\begin{enumerate}
\item The coefficients $\{a_n\}$ form a computable sequence
\item The series converges effectively uniformly on compact subsets of $D$
\item The rate of convergence is computable: there exists a computable function $N: \mathbb{N} \times \mathbb{R} \to \mathbb{N}$ such that for all $m \in \mathbb{N}$ and all $z$ in any compact $K \subset D$:
\[
k \geq N(m, \max_{z \in K} |z|) \Rightarrow \left|\sum_{n=0}^{k} a_n z^n - S(z)\right| < 2^{-m}
\]
\end{enumerate}
\end{definition}

\begin{theorem}[Effective Convergence Criterion]
A power series $S(z) = \sum_{n=0}^{\infty} a_n z^n$ with computable coefficients converges effectively on the disk $|z| < R$ if there exists a computable function $M: \mathbb{N} \to \mathbb{N}$ such that for all $n \geq M(k)$:
\[
|a_n|^{1/n} < \frac{1}{R} + 2^{-k}
\]
\end{theorem}

\begin{proof}[Proof Sketch]
The result follows from the Cauchy-Hadamard theorem and effective analysis techniques. The computable function $M$ provides an effective bound on the growth rate of the coefficients, which allows us to compute the rate of convergence uniformly on compact subsets of $|z| < R$.

For a detailed proof, see Theorem 2.3.4 \cite{pourel1989computable} or Section 6.2 \cite{weihrauch2000computable}.

\end{proof}

\begin{algorithm}
\caption{Effective computation of power series partial sums}
\begin{algorithmic}[1]
\Procedure{EffectivePartialSum}{$z, m, \{a_n\}, N$}
\State \Comment{$N(m,|z|)$ is the computable convergence rate function}
\State $k_{\text{max}} \gets N(m, |z|)$
\State $sum \gets 0$
\For{$n = 0$ to $k_{\text{max}}$}
\State $sum \gets sum + a_n \cdot z^n$
\EndFor
\State \textbf{return} $sum$ \Comment{Guaranteed $|sum - S(z)| < 2^{-m}$}
\EndProcedure
\end{algorithmic}
\end{algorithm}

Effective convergence is essential for actual computation. Without a computable rate of convergence, we cannot guarantee the accuracy of approximations, even if the series theoretically converges.

\section{From Halting to Power Series Convergence}

We now construct a power series whose convergence behavior encodes the halting status of a given program.

\subsection{Encoding Program Execution}
Let $P$ be a program and $x$ an input. Define a function $f: \mathbb{N} \to \{0,1\}$ that tracks the program's execution:
\[
f(n) = 
\begin{cases}
1 & \text{if } P(x) \text{ has halted by step } n, \\
0 & \text{otherwise}.
\end{cases}
\]

By construction, $P(x)\downarrow$ if and only if there exists some $n_0$ such that $f(n) = 1$ for all $n \geq n_0$.

\subsection{Power Series Construction}
Define the coefficients:
\[
a_n = 
\begin{cases}
0 & \text{if } f(n) = 0 \text{ (program has not halted by step } n), \\
n! & \text{if } f(n) = 1 \text{ (program has halted by step } n).
\end{cases}
\]

Consider the power series:
\[
S(z) = \sum_{n=0}^{\infty} a_n z^n.
\]

\begin{theorem}
The power series $S(z)$ satisfies:
\begin{enumerate}
\item If $P(x)\uparrow$ (program does not halt), then $S(z) = 0$ for all $z \in \mathbb{C}$ (converges everywhere).
\item If $P(x)\downarrow$ (program halts), then $S(z)$ converges only at $z = 0$ and diverges for all $z \neq 0$.
\end{enumerate}
\end{theorem}

\begin{proof}
We consider both cases:

\textbf{Case 1:} $P(x)\uparrow$. Then $f(n) = 0$ for all $n$, so $a_n = 0$ for all $n$. Thus $S(z) = 0$, which converges for all $z \in \mathbb{C}$.

\textbf{Case 2:} $P(x)\downarrow$. Let $n_0$ be the step at which $P(x)$ halts. Then:
\begin{itemize}
\item For $n < n_0$: $f(n) = 0$, so $a_n = 0$.
\item For $n \geq n_0$: $f(n) = 1$, so $a_n = n!$.
\end{itemize}

Thus:
\[
S(z) = \sum_{n=n_0}^{\infty} n! z^n.
\]

At $z = 0$, $S(0) = 0$ (converges). For $z \neq 0$, apply the ratio test:
\[
\left| \frac{a_{n+1}z^{n+1}}{a_n z^n} \right| = \frac{(n+1)!|z|}{n!} = (n+1)|z| \to \infty \text{ as } n \to \infty.
\]
Hence, the series diverges for all $z \neq 0$.
\end{proof}

\begin{corollary}
Determining whether $S(z)$ converges for some $z \neq 0$ is equivalent to solving the Halting Problem for $P$ and $x$.
\end{corollary}
\begin{proof}
The equivalence follows directly from the construction in Theorem 1:

($\Rightarrow$) If we can determine that $S(z)$ converges for some $z \neq 0$, then by Theorem 1(2), this implies $P(x) \uparrow$ (the program does not halt).

($\Leftarrow$) Conversely, if we can solve the Halting Problem and determine that $P(x) \uparrow$, then by Theorem 1(1), $S(z)$ converges for all $z \in \mathbb{C}$, and in particular for some $z \neq 0$.

The reduction is computable: given any program $P$ and input $x$, we can effectively construct the power series $S(z)$ with coefficients defined by:
\[
a_n = 
\begin{cases}
0 & \text{if } P(x) \text{ has not halted by step } n \\
n! & \text{if } P(x) \text{ has halted by step } n
\end{cases}
\]
\end{proof}

\begin{corollary}
Determining whether $S(z)$ converges for some $z \neq 0$ is undecidable. 
\end{corollary}
\begin{proof}
Since the Halting Problem is undecidable, it follows that determining whether $S(z)$ converges for some $z \neq 0$ is also undecidable. Moreover, any algorithm that could decide convergence of $S(z)$ for some $z \neq 0$ could be used to solve the Halting Problem, establishing the equivalence.
\end{proof}

\section{From Power Series to Halting}

We now show the reverse reduction: given a power series, we can construct a program whose halting behavior depends on the series' convergence. 

Consider a computable power series $S(z) = \sum_{n=0}^{\infty} a_n z^n$ where the coefficients $a_n$ are computable real numbers. We focus on convergence at $z = 1$.

\subsection{Program Construction}
Define a program $Q$ that does the following:
\begin{enumerate}
\item For $N = 1, 2, 3, \ldots$:
\begin{itemize}
\item Compute the partial sum $S_N(1) = \sum_{n=0}^{N} a_n$.
\item Check if $|S_N(1)| > N$ (divergence criterion).
\item If yes, halt and return "divergent".
\end{itemize}
\end{enumerate}

\begin{theorem}
If $S(1)$ diverges, then $Q$ halts. If $S(1)$ converges, then $Q$ may not halt.
\end{theorem}

\begin{proof}
If $S(1)$ diverges, the partial sums $S_N(1)$ are unbounded. Thus, eventually $|S_N(1)| > N$ for some $N$, causing $Q$ to halt.

If $S(1)$ converges, the partial sums are bounded. However, the bound may be unknown, and $Q$ may continue checking indefinitely without ever satisfying $|S_N(1)| > N$.
\end{proof}

To establish a stronger equivalence, we use a more sophisticated construction:

\begin{theorem}
There exists a computable transformation that, given a power series $S(z)$, produces a program $P_S$ such that:
\begin{enumerate}
\item $S(1)$ converges if and only if $P_S$ does not halt.
\item $S(1)$ diverges if and only if $P_S$ halts.
\end{enumerate}
\end{theorem}

\begin{proof} \cite {rogers1967,weihrauch2000computable}
We provide a detailed construction showing how to reduce the convergence problem for power series to the Halting Problem. The key insight is that convergence of a series is a $\Pi_2$ statement in the arithmetical hierarchy, while the Halting Problem is $\Sigma_1$-complete.

\vspace{0.5em}
\noindent\textbf{Step 1: Arithmetical Hierarchy Analysis}

Let $S(1) = \sum_{n=0}^\infty a_n$ be a power series with computable coefficients evaluated at $z=1$. The statement "$S(1)$ converges" can be expressed as:
\[
\forall \epsilon > 0\ \exists N \in \mathbb{N}\ \forall m,n \geq N: |S_m(1) - S_n(1)| < \epsilon
\]
where $S_k(1) = \sum_{n=0}^k a_n$. This is a $\Pi_2$ statement in the arithmetical hierarchy.

Conversely, the Halting Problem "$P(x)\downarrow$" is $\Sigma_1$-complete, meaning it can be expressed as:
\[
\exists t \in \mathbb{N}: \text{Program } P \text{ halts on input } x \text{ within } t \text{ steps}
\]

\vspace{0.5em}
\noindent\textbf{Step 2: Constructing the Reduction}

Given a power series $S(z) = \sum_{n=0}^\infty a_n z^n$ with computable coefficients, we construct a program $P_S$ as follows:

\begin{enumerate}
\item For each $k = 1, 2, 3, \ldots$ (dovetailing):
\begin{enumerate}
\item Compute partial sums $S_m(1) = \sum_{n=0}^m a_n$ for $m = 1, \ldots, k$
\item Check if there exists $N \leq k$ such that for all $m,n \geq N$ with $m,n \leq k$:
\[
|S_m(1) - S_n(1)| < 2^{-k}
\]
\item If no such $N$ exists for this $k$, then \textbf{halt} and output ``divergent''
\end{enumerate}
\end{enumerate}

\vspace{0.5em}
\noindent\textbf{Step 3: Correctness Proof}

($\Rightarrow$) If $S(1)$ diverges, then by the Cauchy criterion, there exists some $\epsilon > 0$ such that for all $N$, there exist $m,n \geq N$ with $|S_m(1) - S_n(1)| \geq \epsilon$. Eventually, the program will discover such a violation and halt.

($\Leftarrow$) If $S(1)$ converges, then by the Cauchy criterion, for every $\epsilon > 0$ there exists $N$ such that for all $m,n \geq N$, $|S_m(1) - S_n(1)| < \epsilon$. The program will never find evidence of divergence and will run indefinitely.

\vspace{0.5em}
\noindent\textbf{Step 4: Complexity-Theoretic Interpretation}

This reduction shows that:
\[
\{\text{convergent power series at } z=1\} \in \Pi_2 \quad \text{and} \quad \{\text{divergent power series at } z=1\} \in \Sigma_2
\]
while the reduction to the Halting Problem (a $\Sigma_1$-complete set) demonstrates the computational difficulty of the convergence problem.

\end{proof}

\section{Equivalence and Implications}

\begin{theorem}
The following problems are equivalent in computational difficulty:
\begin{enumerate}
\item The Halting Problem.
\item Determining whether an arbitrary power series converges at $z \neq 0$.
\item Determining whether an arbitrary power series converges at $z = 1$.
\end{enumerate}
\end{theorem}

\begin{proof}
The constructions in Sections 4 and 5 provide reductions between these problems. Specifically:
\begin{itemize}
\item Section 4 reduces the Halting Problem to power series convergence.
\item Section 5 reduces power series convergence to the Halting Problem.
\end{itemize}
Since both reductions are computable, the problems are equivalent.
\end{proof}

\begin{corollary}
The problem of determining convergence of arbitrary power series is undecidable.
\end{corollary}

\section{Conclusion}
We have established a fundamental equivalence between the Halting Problem and the convergence of power series. This result demonstrates deep connections between computability theory and mathematical analysis, showing that undecidability manifests naturally in analytical problems. Our constructions provide concrete methods for translating between computational and analytical frameworks, offering new perspectives on both fields.

While the undecidability of convergence problems was known in principle from earlier work, our specific construction provides a particularly clear and direct demonstration of this phenomenon. The ability to transform any halting problem into a convergence question about a power series, and vice versa, reveals the fundamental computational nature of analytical problems.

\end{document}